\documentclass[journal, 10pt, twocolumn, twoside]{IEEEtran}

\makeatletter
\def\ps@headings{%
\def\@oddhead{\mbox{}\scriptsize\rightmark \hfil \thepage}%
\def\@evenhead{\scriptsize\thepage \hfil \leftmark\mbox{}}%
\def\@oddfoot{}%
\def\@evenfoot{}}
\makeatother
\usepackage{amssymb,amsmath,color,graphicx}
\usepackage{subfigure}
\usepackage{cite}
\usepackage{verbatim}
\usepackage{algorithm}
\usepackage{algpseudocode}
\usepackage{multirow}
\everymath{\displaystyle}
\allowdisplaybreaks

\newtheorem{theorem}{Theorem}[section]

\newtheorem{lemma}{Lemma}[section]

\newtheorem{proposition}{Proposition}[section]
\newtheorem{remark}[theorem]{Remark}

\ifodd 0
\newcommand{\revbegin}{\color{red}} 
\newcommand{\revend}{\color{black}} 
\newcommand{\rev}{\color{red}} 
\newcommand{\revno}{\color{black}} 
\else
\newcommand{\revbegin}{} 
\newcommand{\revend}{} 
\newcommand{\rev}{} 
\newcommand{\revno}{} 
\fi

\begin{document}

\title{Spatial Performance Analysis and Design Principles for Wireless Peer Discovery}

\author{Taesoo Kwon,~\IEEEmembership{Member,~IEEE}
        and Ji-Woong Choi,~\IEEEmembership{Senior Member,~IEEE}

\thanks{T. Kwon is with the Electronics and Telecommunications Research Institute (ETRI), Daejeon, 305-700, South Korea (e-mail: tskwon@etri.re.kr).}
\thanks{J.-W. Choi is with Department of Information \& Communication Engineering, Daegu Gyeongbuk Institute of Science and Technology (DGIST), Daegu, 771-873, South Korea (e-mail: jwchoi@dgist.ac.kr).}

}

\markboth{Kwon and Choi: Spatial Performance Analysis and Design Principles for Wireless Peer Discovery}
{Kwon and Choi: Spatial Performance Analysis and Design Principles for Wireless Peer Discovery}

\maketitle

\begin{abstract}
  In wireless peer-to-peer networks that serve various proximity-based applications, peer discovery is the key to identifying other peers with which a peer can communicate and an understanding of its performance is fundamental to the design of an efficient discovery operation.
  This paper analyzes the performance of wireless peer discovery through comprehensively considering the wireless channel, spatial distribution of peers, and discovery operation parameters.
  The average numbers of successfully discovered peers are expressed in closed forms for two widely used channel models, i.e., the interference limited Nakagami-$m$ fading model and the Rayleigh fading model with nonzero noise, when peers are spatially distributed according to a homogeneous Poisson point process.
  These insightful expressions lead to the design principles for the key operation parameters including the transmission probability, required amount of wireless resources, level of modulation and coding scheme (MCS), and transmit power.
  \revbegin Furthermore, the impact of shadowing on the spatial performance and suggested design principles is evaluated using mathematical analysis and simulations. \revend
\end{abstract}

\begin{IEEEkeywords}
  Peer discovery, neighbor discovery, stochastic geometry, D2D networks, random access protocol.
\end{IEEEkeywords}

\section{Introduction} \label{sec:Intro}

  Recently, it is considered that wireless peer-to-peer communications will enable novel and significant opportunities such as proximal social networking, network offloading, and public safety \cite{LTE_ProSe}.
  Accordingly both industrial and academic communities have begun to increasingly investigate the potential new services and technical challenges \cite{Fodor12_D2DmagEricsson,Zou13_P2P_Discovery}.
  For wireless peer networking, each peer should first be able to identify other peers with which it can communicate before transmitting and receiving data.
  This operation is referred to as \emph{peer discovery}, which is the most basic process for establishing connections and building topology information in various wireless networks including device-to-device (D2D) networks and sensor networks.
  However, the performance of peer discovery is significantly affected by the randomness of the wireless channel as well as peer location.
  The primary focus of this paper is to quantify the implications of the wireless channel and spatial distribution of peers on wireless peer discovery and to derive design principles from the results.

  Even though peer discovery is fundamental to the operation of wireless networks, wireless resources for this process are a control overhead that does not contribute to increasing data capacity.
  In this regard, peer discovery should be designed to find as many peers as possible using a small amount of wireless resources, in order to minimize the overhead.
  In such a design, wireless resources for peer discovery should be spatially shared among peers, and this spatial reuse results in performance degradation due to interference signals.
  In this sense, understanding the effect of the interference signals from spatially distributed peers is the key to designing efficient peer discovery schemes.
  Recent studies have attempted to statistically model a wireless network topology using the mathematical tool of stochastic geometry \cite{Baccelli09_Aloha,Haenggi09_SG_Survey,Andrews10_SG_SpatialModelingCM}: this model facilitates the derivation of the spatial probability distribution of the signal to interference and noise ratio (SINR).
  This paper investigates wireless peer discovery based on this stochastic geometry theory.


\subsection{Related Work} \label{ssec:RelatedWork}

  Several studies have suggested aggressive schemes where each peer can transmit its unique signal and simultaneously receive multiple signatures from other peers for rapid and collision-free peer discovery, e.g., \cite{Lin04_SubspaceNbrDiscovery,Guo08_Luo_NbrGroupTest}, but a simple random access protocol is still regarded as the basis of wireless peer discovery \cite{McGlynn01_NbrDiscBirthday,Borbash07_AsyncNbrDisc,Vasudevan09_NbrDiscCoupon,Vigato11_JointDiscovery,Zeng12_Nbr_Multipacket,Jeon12_NeighborDiscovery,Hamida08_NbrDiscovery,Bacelli12_SG_D2Ddiscovery} because the lack of a priori information about peers in dynamic wireless networks may only provide the uncoordinated sharing of peer discovery resources among peers.

  The primary reason for performance degradation in a random access protocol is packet collisions due to the simultaneous transmission of peers; thus, several studies have investigated the quantification and improvement of the peer discovery performance based on the packet collision model \cite{McGlynn01_NbrDiscBirthday,Borbash07_AsyncNbrDisc,Vasudevan09_NbrDiscCoupon}.
  However, this collision model oversimplifies wireless receiving operations.
  In fact, the success or failure of packet reception is primarily determined by the physical layer metrics, e.g., SINR, rather than whether or not packets simply collide.
  In addition, the requirement of this received SINR depends on the physical transceiving scheme, such as the receiver structure and level of modulation and coding scheme (MCS).
  Based on this, there have been attempts to understand the effect of physical layer characteristics including the receiver structure and wireless channel \cite{Vigato11_JointDiscovery,Zeng12_Nbr_Multipacket}.
  In \cite{Vigato11_JointDiscovery}, a joint iterative decoding method for multiuser detection was applied to peer discovery but its system performance improvement was only evaluated using simulation.
  The work of \cite{Zeng12_Nbr_Multipacket} analyzed the performance of multipacket reception based on the conventional packet collision model.
  These approaches remains insufficient to reveal the implications of the randomness of wireless channels and peer locations.

  There have also been recent studies that analyze the spatial performance by statistically modeling peer location and wireless channel \cite{Jeon12_NeighborDiscovery,Hamida08_NbrDiscovery,Bacelli12_SG_D2Ddiscovery}.
  In \cite{Jeon12_NeighborDiscovery}, the received power of the signals from randomly located peers was modeled in a probabilistic manner and the multipacket reception capability was assumed.
  However, the performance was only expressed in a form with as many cumbersome integrations as the number of peers; therefore, this result could not explicitly present the design implications of wireless peer discovery.
  In contrast, the works of \cite{Hamida08_NbrDiscovery,Bacelli12_SG_D2Ddiscovery} attempted to mathematically analyze the peer discovery performance with interference considerations using a stochastic geometry framework \cite{Baccelli09_Aloha,Haenggi09_SG_Survey,Andrews10_SG_SpatialModelingCM}.
\revbegin
  The work of \cite{Hamida08_NbrDiscovery} compared the packet collision and SINR models when peers were distributed according to a homogeneous Poisson point process (PPP), and it expressed the average number of discovered peers in a closed form under the Rayleigh fading channel when the noise power can be ignored.
  A similar result was also presented in \cite{Bacelli12_SG_D2Ddiscovery}.
  The results under this specific channel model, i.e., the Rayleigh fading with zero noise power, provide a basis for the analysis under channel models that belong to the exponential family, e.g. the Nakagami-$m$ fading channel \cite{ElSawy13_StochGeomSurvey}.
  However, the explicit derivation of the performance under more general channel models, such as incorporating the Nakagami-$m$ fading channel, nonzero background noise power, and shadowing, also has the significant merit because it enables the clarification of the relationship between wireless channels and discovery operation parameters, which leads to the design principles for the key operation parameters including the transmission probability for a half duplex operation, received SINR requirement, and transmit power under various channel environments.
  This extension was not considered in \cite{Hamida08_NbrDiscovery} and \cite{Bacelli12_SG_D2Ddiscovery}.
\revend

\subsection{Contributions and Organization} \label{ssec:Contributions}

  This paper investigates wireless peer discovery with respect to the mean number of successfully discovered peers, which is denoted by $\mathbb{E}\{S\}$, by comprehensively considering the wireless properties as well as the discovery operation properties.
  The main contributions are highlighted into the following three aspects.

  \subsubsection{Deriving the Average Number of Successfully Discovered Peers}
  The closed forms for $\mathbb{E}\{S\}$ are derived for two widely used channel models: (i) the interference limited Nakagami-$m$ fading model and (ii) the Rayleigh fading model with nonzero noise power.
\revbegin
  These elegant expressions comprehensively quantify the effect of the wireless channels, spatial peer distribution, and operation parameters.
  In particular, these results clarify the impact of the Nakagami-$m$ fading channel and noise power, unlike prior studies \cite{Hamida08_NbrDiscovery,Bacelli12_SG_D2Ddiscovery} that have only derived the closed form expression of $\mathbb{E}\{S\}$ under the Rayleigh fading channel with zero noise power.
  For example, the mathematical analysis reveals that $\mathbb{E}\{S\}$ is independent of the Nakagami-$m$ fading parameter (i.e., $m$) under an interference limited environment where the aggregate interference overwhelms the noise power.
\revend

  \subsubsection{Suggesting Design Principles for Discovery Operation Parameters}
  The design of optimal or suboptimal discovery operation parameters is investigated in terms of maximizing $\mathbb{E}\{S\}$ under the two channel models mentioned above.
  An important difference between the two models is the noise; it is demonstrated that this difference may result in significantly different design principles for the parameters.
  For example, regarding the transmission probability for a half duplex operation that is denoted by $\rho$, $\mathbb{E}\{S\}$ increases as $\rho$ decreases under an interference limited environment, whereas it becomes a unimodal function of $\rho$ when the noise power cannot be ignored.
  The insightful results derived in this paper are summarized in Table~\ref{table:Summary}.

  \subsubsection{{\revno Evaluating Performance under Various Channel Models}}
\revbegin
  The analytical results for the two channel models without shadowing are extended to those under channel models that incorporate general shadowing through applying the displacement theorem, similar to the work presented in \cite{Dhillon13_DlHetnetGen}.
  This extension reveals that the performance under the interference-limited scenario is invariant to the shadowing distribution; however, for a nonzero noise power, the shadowing tends to reduce the impact of the noise power.
  In contrast, the analytical results do not embrace wireless channel models that incorporate all of the general path loss exponent, Nakagami-$m$ fading, nonzero noise power, and shadowing.
  In order to fill this void, simulations are used to demonstrate that the performances under such general channel models are consistent with those derived analytically under the specific channel models.
\revend

  The remainder of this paper is organized as follows.
  Section~\ref{sec:HelloProtocol} describes the system model for a multichannel random hello protocol and presents the spatial performance under a general wireless channel model in terms of the average number of successfully discovered peers.
  Sections~\ref{sec:ZeroNoise} and \ref{sec:NonzeroNoise} analyze the spatial performance of the peer discovery protocol and suggest design principles for discovering as many peers as possible, under the interference limited Nakagami-$m$ fading channel and the Rayleigh fading channel with nonzero noise, respectively.
\revbegin
  Section~\ref{sec:Shadowing} extends the results derived in the previous two sections into those for wireless channel models that incorporate arbitrary shadowing.
\revend
  Then, Section~\ref{sec:Results} discusses numerical results, and Section~\ref{sec:Conclusions} concludes the paper.

\section{Multichannel Random Hello Protocol for Wireless Peer Discovery} \label{sec:HelloProtocol}
%
%
%

\subsection{System Model} \label{ssec:Model}

  This paper considers the multichannel random hello protocol for wireless peer discovery illustrated in Fig.~\ref{fig:HelloProtocolModel} when peers or nodes\footnote{In this paper, both the terms are used synonymously.} are randomly distributed in a two-dimensional space.
\revbegin
  The model assumes resource orthogonality, i.e., signals transmitted over different resources do not interfere with each other.
  The premise for this orthogonality is global synchronization \cite{Corson10_ProximityInterworking}.
  If nodes are not precisely time synchronized, the interference that results from time mismatches may significantly degrade the performance.
  However, in general, the time synchronization in distributed wireless networks is a resource and energy intensive task.
  External signals from the existing infrastructure, e.g. nearby cellular base stations, may render this difficult task simpler \cite{Corson10_ProximityInterworking,Wu13_FlashLinQ}; such signals are exploited as timing reference signals for rough synchronization, and then the remaining time offsets are further corrected through additional synchronization procedures among the nodes in order that the residual timing errors can be readily accommodated in the signal level, e.g., using cyclic prefix (CP) in OFDMA systems.
  Accordingly, this paper assumes that all nodes are time synchronized and does not consider performance degradation due to time mismatches.
\revend
  The model also assumes that all nodes operate in a half duplex manner.
  For this half duplex operation, a node decides whether it transmits or receives a hello packet identifying a node every time slot in a probabilistic manner, and $\rho$ denotes the probability that a node transmits in a time slot, i.e., the \emph{transmission probability}.

\begin{figure}[t]
\centering
\includegraphics[width=9cm]{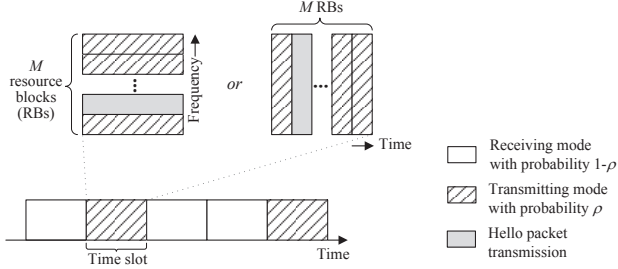}
\vspace{-0.59cm}
\caption{Multichannel random hello protocol for wireless peer discovery.}
\label{fig:HelloProtocolModel}
\end{figure}

  One time slot consists of $M$ resource blocks (RBs) in a frequency or time domain, and one RB is used for a hello packet transmission.
  For peer discovery, nodes in the transmitting mode broadcast their hello packet using one RB randomly chosen among $M$ RBs while nodes in the receiving mode try to detect the packets simultaneously over all $M$ RBs in a time slot.
  Let $\xi > 0$ denote the minimum received SINR required for the successful reception of a packet.
  If a node in the receiving mode receives a hello packet with SINR above $\xi$, then it means that this node successfully discovers the node that transmits the packet.
  Note that the value of $\xi$ determines the MCS level at which a hello packet is transmitted.

  The model assumes that the nodes are spatially distributed according to a homogeneous PPP with node density $\lambda$, which is denoted by $\Phi$.
  In order to investigate the node average performance, the performance of a reference receiving node is observed and such a node is referred to as \emph{a typical node}.
  A typical node is assumed to be located at the origin and search potential target node $i$ in a time slot.
  If node $i$ is transmitting in the same time slot, the signal transmitted by node $i$ becomes the desired signal of a typical node.
  Assume that node $i$ transmits a hello packet using the $m$th RB in the time slot.
  Under these assumptions, all signals sent over the $m$th RB by nodes in the transmitting mode, other than node $i$, become interference.
\revbegin
  Note that a typical node is interested in hello packets from all other nodes; thus, target node $i$ indicates an arbitrary node rather than a specific node.
  Therefore, according to Slivnyak's theorem \cite{Stoyan96_StochasticGeom}, nodes except a typical node and the target node $i$ still constitute a homogeneous PPP with the same density as $\lambda$.
\revend
  These nodes are potential interferers.
  Let $\Phi_{q}$ denote a homogeneous PPP with density $\lambda q$ that results from the independent thinning of homogeneous PPP $\Phi$ with retention probability $q$.
  In a given time slot, each node is transmitting with probability $\rho$ and it uses the same RB as that of node $i$ with probability $1/M$.
  Thus, the spatial distribution of the interfering nodes can be modeled as the thinning of an original PPP with a retention probability $\rho/M$, and it is expressed as a homogeneous PPP with density $\lambda \rho / M$, i.e., $\textstyle \Phi_{\rho/M}$.
\revbegin
  In fact, each node in $\textstyle \Phi_{\rho/M}$ becomes both a potential target node and a interferer of a typical node.
\revend
  $X_j$ denotes the location of node $j$ and $|X_j|$ represents the distance from the origin to $X_j$.
  Assume that all nodes have the same transmit power $p$.
  The standard power loss propagation model with the path loss exponent of $\alpha$ ($>2$) is supposed.
  Let $h_{i}$ and $g_{j}$ denote the fading power gains\footnote{They may denote the channel power gain including the shadowing as well as the Nakagami-$m$ or Rayleigh fading, depending on the wireless channel model.} that the desired signal from node $i$ and the interfering signal from node $j$ undergo, respectively.
  It is assumed that $\{h_{i}\}$ and $\{g_{j}\}$ are independently and identically distributed (i.i.d.), respectively.
  Fig.~\ref{fig:SpatialModel} explains the spatial model considered in this paper.

%
\begin{figure}[t]
\centering
\includegraphics[width=9cm]{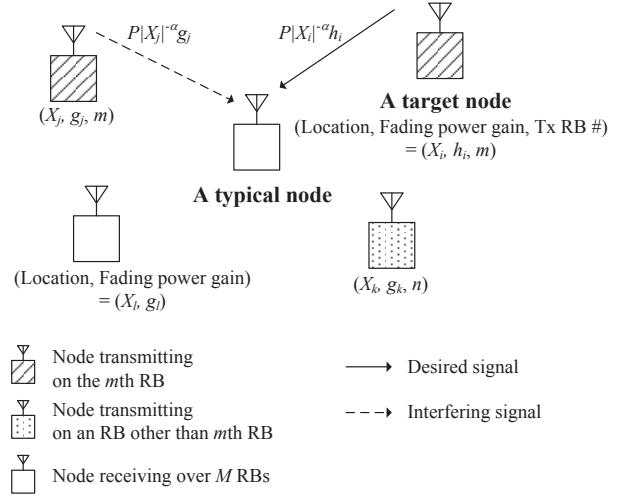}
\vspace{-0.5cm}
\caption{Spatial model for performance analysis of wireless peer discovery.} \label{fig:SpatialModel}
\end{figure}

\revbegin
  Eventually, when $\Xi(X_i)$ denotes the received SINR at a typical node for the hello packet transmitted by a target node located on $X_i$,
%
\begin{IEEEeqnarray}{ll}\label{eq:Sinr}
  \Xi(X_i) & = \frac{|X_i|^{-\alpha} h_i}{\sum_{j \in \Phi_{\rho/M}} |X_j|^{-\alpha} g_j + \sigma^2},
\end{IEEEeqnarray}
\revend
  where $\textstyle \sigma^2 \triangleq \frac{\tilde{\sigma}^2}{p}$ and $\tilde{\sigma}^2$ denotes the noise power.
  Herein, $\textstyle \frac{1}{\sigma^2} = \frac{p}{\tilde{\sigma}^2}$ can be understood as the average received signal to noise ratio (SNR) at a unit distance, i.e., when $|X_i|=1$.

\begin{table*}[!t]
\caption[Caption for LOF]{The property of the average number of successfully discovered nodes, $\mathbb{E}\{S\}$\protect\footnotemark.}
\centering
\label{table:Summary}
    \begin{tabular}{|c||c|c|}
    \hline
        & Interference limited case ($\sigma^2=0$)
        & Nonzero noise power case ($\sigma^2>0$)
    \\ \hline\hline
        Path loss exponent, $\alpha$
        & Increasing with $\alpha$
        & Increasing with $\alpha$$^\dag$ 
    \\ \hline
        Nakagami-$m$ fading parameter, $m$
        & Independent of $m$
        & Insignificant$^\dag$
    \\ \hline
        Node density, $\lambda$
        & Independent of $\lambda$
        & Increasing with $\lambda$, but saturated
    \\ \hline
        Number of RBs, $M$ (for a fixed $\xi$)
        & Linearly increasing with $M$
        & Increasing with $M$, but saturated
    \\ \hline
        SINR threshold, $\xi$
        & Unimodal function of $\xi$
        & Maybe, unimodal function of $\xi$$^\dag$
    \\ \hline
        Transmission probability, $\rho$
        & Increasing as $\rho$ decreases
        & Unimodal function of $\rho$
    \\ \hline
        Transmit power, $p$
        & Independent of $p$
        & Increasing with $p$, but saturated
    \\ \hline
        \rev
        Standard deviation of lognormal shadowing, $\chi$
        & \rev Independent of $\chi$
        & \rev Increasing with $\chi$, but saturated
    \\ \hline
    \end{tabular}
\end{table*}
%

\subsection{Spatial Performance Metric} \label{ssec:Performance}

  In this paper, wireless peer discovery aims to find as many nodes as possible, i.e., to maximize the average number of successfully discovered nodes.
  The successful peer discovery requires the fulfillment of the following three conditions:
    (i) a typical node is in the receiving mode;
    (ii) a target node is in the transmitting mode;
    and (iii) a hello packet that the target node transmits should be received with an SINR above $\xi$ at a typical node.
  The status of node $j$ is represented by $Z_j$, i.e., $Z_j=0$ if node $j$ is in the receiving mode and $Z_j=1$ otherwise.
  Let $P(X|Z_0=0)$ denote the probability that a typical node successfully discovers a target node located on $X$ when a typical node is in the receiving mode.
  Then, $P(X|Z_0=0)$ is given by
\begin{IEEEeqnarray}{ll}\label{eq:PxDef}
  P(X|Z_0=0) = \rho\Pr\left\{ \Xi(X) > \xi \right\}.
\end{IEEEeqnarray}
  This success probability depends on the wireless channels and spatial distribution of the nodes, thus $P(X|Z_0=0)$ is expressed as a function of wireless channel parameters and node density $\lambda$ as well as discovery operation parameters including $M$, $\xi$, $\rho$, and $p$.

\revbegin
  This paper considers wireless channel models that embrace shadowing as well as Nakagami-$m$ fading.
  That is, the fading power gain is given by the product of the gains that result from the shadowing and Nakagami-$m$ fading.
  The Nakagami-$m$ fading model does not only generalize or approximate various useful fading channels such as the Rayleigh and Rician fading channels, but it also allows a closed form of $\mathbb{E}\{S\}$ in some specific cases, e.g., when $\sigma^2=0$, where $\mathbb{E}\{S\}$ denotes the average number of nodes that a typical node successfully discovers over $M$ RBs.
  This is derived in the following sections.
  In addition, the results yielded in this Nakagami-$m$ fading model can be readily extended to those for wireless channel models that incorporate shadowing as long as the shadowing of links is i.i.d. \cite{Dhillon13_DlHetnetGen,Madhusudhanan12_GenScs}.
  Therefore, the mathematical analysis in Sections~\ref{sec:ZeroNoise} and \ref{sec:NonzeroNoise} concentrates on the Nakagami-$m$ fading and Rayleigh fading without shadowing.
  Then, the results will be extended to those for more general wireless channel models that incorporate shadowing.
\revend

  When the Nakagami-$m$ fading model is only considered, the fading power gain $h$ follows the Gamma distribution and its complementary cumulative distribution function (ccdf) is given as follows:
\begin{IEEEeqnarray}{l}\label{eq:NakaCcdf}
  \Pr\{h > x\} = \exp(-mx) \sum_{k=0}^{m-1} \frac{m^k}{k!} x^k,
\end{IEEEeqnarray}
  where $m$ is the Nakagami-$m$ fading parameter.
  Under this Nakagami-$m$ fading channel model, the following lemma is a start toward deriving a simple form of  $\mathbb{E}\{S\}$.
%
\begin{lemma}\label{lem:SnakaNzN}
  When the desired and interfering signals undergo the Nakagami-$m$ fading with $m_s$ and $m_i$, which are positive integers,
  the average number of nodes that a typical node successfully discovers is given as follows:
\begin{IEEEeqnarray}{ll}\label{eq:EsNakaNzN}
  \hspace{-0.2cm} \mathbb{E}\{S\} = 2 \pi \lambda \rho(1-\rho) \sum_{k=0}^{m_s - 1} \frac{(-m_s \xi)^k}{k!} \int_{0}^{\infty} r^{k\alpha + 1} \cdot \IEEEnonumber\\
  \hspace{0.5cm}
  \left. \frac{d^k \exp\left( -\frac{\lambda\rho}{M} \pi \zeta^{\frac{2}{\alpha}} \Delta_i(m_i,\alpha) - \zeta\sigma^2 \right)}{d\zeta^k} \right|_{\zeta=m_s \xi r^\alpha} dr, 
\end{IEEEeqnarray}
%
  where $\textstyle \Delta_i(m_i,\alpha) \triangleq m_i^{-\frac{2}{\alpha}} \frac{\Gamma\left(1-\frac{2}{\alpha}\right) \Gamma\left(m_i+\frac{2}{\alpha}\right)}{\Gamma(m_i)}$ and $\textstyle \Gamma(x)\triangleq\int_{0}^{\infty} t^{x-1}\exp(-t)dt$ denotes the Gamma function.
\end{lemma}
\begin{proof}
  See Appendix~\ref{app:proof:lem:SnakaNzN}.
\end{proof}
\footnotetext{The superscript $\dag$ denotes that the observation was from simulation results. All the others are mathematically demonstrated. Regarding the Nakagami-$m$ fading parameter, $m_s=m_i=m$ is assumed.}
  Lemma~\ref{lem:SnakaNzN} requires the integrations of higher order derivative terms, and
\revbegin
  it remains difficult to express the result in a closed form.
\revend
  However, it is noteworthy that $\mathbb{E}\{S\}$ in \eqref{eq:EsNakaNzN} can be expressed in more elegant forms by imposing environmental constraints.
  Therefore, this paper focuses more on two specific but widely used channel models: (i) the Nakagami-$m$ fading with $\sigma^2=0$ and $\alpha>2$, and (ii) the Rayleigh fading with $\sigma^2>0$ and $\alpha=4$.
  The results for these models do not only quantify the effect of wireless channels but also offer useful design principles for the crucial operation parameters of $M$, $\xi$, $\rho$, and $p$.
  The following two sections discuss these two channel models, and Table~\ref{table:Summary} summarizes the main results derived in this paper.

\section{Spatial Analysis and Design Principles for Interference Limited Channels} \label{sec:ZeroNoise}

  This section investigates the interference limited case, which is modeled as $\sigma^2=0$ in Lemma~\ref{lem:SnakaNzN}.
  The results remain general in terms of the Nakagami-$m$ fading parameter and path loss exponent ($\alpha>2$).
  Under this wireless environment, the design principles for wireless peer discovery are suggested by deriving the values of $M$, $\xi$, and $\rho$ for maximizing $\mathbb{E}\{S\}$.

\subsection{Spatial Performance}\label{ssec:EsNakaZn}

  When $\sigma^2=0$, from Lemma~\ref{lem:SnakaNzN}, the following results are obtained.
\begin{proposition}\label{prop:EsNakaZn}
  When the desired and interference signals undergo the Nakagami-$m$ fading with $m_s$ and $m_i$, respectively, and $\sigma^2=0$, $\mathbb{E}\{S\}$ is given by
\begin{IEEEeqnarray}{l}\label{eq:EsNakaZn}
  \mathbb{E}\{S\} = \frac{\Delta_s\left(m_s,\alpha \right)}{\Delta_i\left(m_i,\alpha \right)} \frac{M(1-\rho)}{\xi^{\frac{2}{\alpha}}},
\end{IEEEeqnarray}
  where $\textstyle \Delta_s (m_s,\alpha) \triangleq m_s^{-\frac{2}{\alpha}}\frac{ \Gamma\left(m_s+\frac{2}{\alpha}\right) }{ \Gamma\left(1+\frac{2}{\alpha}\right)\Gamma\left(m_s\right) }$.
  In particular, if $m_s = m_i$, then
\begin{IEEEeqnarray}{l}\label{eq:EsNakaZnM}
  \mathbb{E}\{S\} = \frac{\sin(2\pi/\alpha)}{2\pi/\alpha} \frac{M(1-\rho)}{\xi^{\frac{2}{\alpha}}},
\end{IEEEeqnarray}
  which is \emph{independent of the Nakagami-$m$ fading parameter}.
\end{proposition}
\begin{proof}
  By calculating \eqref{eq:EsNakaNzN} for $\sigma^2=0$, \eqref{eq:EsNakaZn} and \eqref{eq:EsNakaZnM} can be obtained.
  For more details, see Appendix~\ref{app:proof:prop:EsNakaZn}.
\end{proof}
%

  Because a typical node attempts to discover \emph{any} node rather than a specific node, it is sensible to assume that the wireless fading channel statistics for the desired and interfering signals are the same, i.e., $m_s=m_i$.
  In this regard, it is quite interesting that $\mathbb{E}\{S\}$ in \eqref{eq:EsNakaZnM} does not depend on the Nakagami-$m$ fading parameter, $m_s$ or $m_i$.
  This result can be interpreted as the fading effects for the desired and interfering signals being counterbalanced in terms of $\mathbb{E}\{S\}$.
  Furthermore, it is noteworthy that $\mathbb{E}\{S\}$ in \eqref{eq:EsNakaZn} is independent of node density $\lambda$.
  When noise is neglected, the node density only affects the geometric size of the wireless networks and this node density does not change the ratio of distances of the target node and interfering nodes from a typical node.
  That is, the increases or decreases in the desired and interfering signal powers according to the node density cancel each other out in terms of $\mathbb{E}\{S\}$, which is similar to the fading parameter.
  In summary, Proposition~\ref{prop:EsNakaZn} signifies that the fading parameter and node density can be considered as unimportant when $\sigma^2=0$ and $m_s=m_i$.
  Now, the focus is moved to the path loss exponent $\alpha$.
  Note that $\textstyle \frac{\sin(x)}{x}$ is a monotonically decreasing function of $0<x<\pi$.
  Thus, $\mathbb{E}\{S\}$ in \eqref{eq:EsNakaZnM} increases with $\alpha$ when $m_s=m_i$.
  The number of nodes further than the target node from a typical node is always significantly more than the closer ones when considering nodes distributed in an infinite two-dimensional space.
  Therefore, the aggregate interference decays more quickly than the power of the desired signal as $\alpha$ increases, and this results in $\mathbb{E}\{S\}$ increasing with $\alpha$.



\subsection{Design of the Discovery Operation Parameters}\label{ssec:DesignNakaZn}

  In \eqref{eq:EsNakaZn}, $\mathbb{E}\{S\}$  is determined by both the wireless channel parameters, e.g., $m_s$, $m_i$, and $\alpha$, and the discovery operation parameters, e.g., $M$, $\xi$, and $\rho$.
  This subsection elaborates on the design of the three discovery operation parameters.
  The design aims to maximize $\mathbb{E}\{S\}$ in \eqref{eq:EsNakaZn}.
  It is trivial to derive the optimal value of $\rho$ for maximizing $\mathbb{E}\{S\}$, which is denoted by $\rho^*$, i.e., $\textstyle \rho^* \triangleq \arg \max_{0<\rho<1} \mathbb{E}\{S\}$.
  $\rho^*$ always approaches zero regardless of $M$ and $\xi$.
  Therefore, the design of $M$ and $\xi$ can be separated from that of $\rho$, when $\sigma^2=0$.
  The meaning of $\rho^*$ will be discussed in more detail at the end of this subsection.

  When $\xi$ is given, $\mathbb{E}\{S\}$ in \eqref{eq:EsNakaZn} increases linearly with $M$.
  That is, regarding maximizing $\mathbb{E}\{S\}$ in \eqref{eq:EsNakaZn}, for a fixed $\xi$, $M$ can be designed regardless of the wireless channel parameters and $\rho$.
  However, the design of $M$ is highly correlated with that of $\xi$.
  As mentioned in Section~\ref{ssec:Model}, the value of $\xi$ determines the MCS level, i.e., the data rate, available for a hello packet transmission.
  When the total amount of peer discovery resources are fixed, the data rate decided by $\xi$ affects how many RBs the total resources can be divided into, thus $M$ and $\xi$ should be jointly designed.

\begin{figure}[t]
\centering
\includegraphics[width=9.3cm]{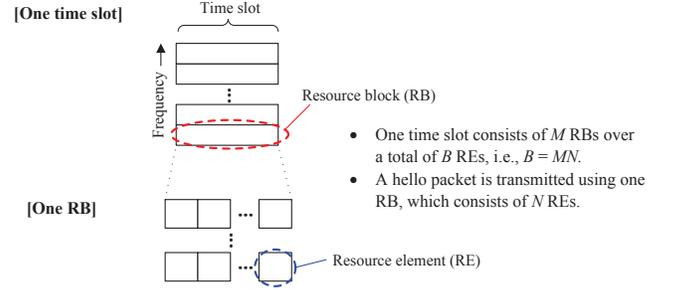}
\vspace{-0.5cm}
\caption{Resource structure for a hello packet transmission (this figure only shows a time slot consisting of multiple RBs in a frequency domain but it can also be divided into RBs in a time domain, similar to the ones of a frequency domain).}
\label{fig:HelloProtocolResource}
\end{figure}

  The joint design begins with quantifying the relationship between $M$ and $\xi$.
\revbegin
  Typically, the data rate for a finite-length packet is always below the Shannon capacity and a hello packet also conveys small size information for identifying a node, e.g., with tens of information bits.
  The SNR gap approximation provides a useful method for representing the SNR or data rate loss with respect to the Shannon capacity \cite{Slepian63_BoundsCommun,Forney98_ModCodLinearGaussian}.
  Accordingly, the data rate is modeled as this SNR gap approximation, i.e, $\textstyle \tau = \log_2(1+\frac{\xi}{\delta})$ bps/Hz where $\delta$ is the SNR gap and positive.
\revend
  It is noteworthy that as $\xi$ increases, the data rate of a hello packet transmission increases while the success probability given by \eqref{eq:PxDef} decreases.
  In order to quantify this tradeoff, first, the resource structure for the multichannel random hello protocol introduced in Fig.~\ref{fig:HelloProtocolModel} is revisited.  
  Fig.~\ref{fig:HelloProtocolResource} provides a more detailed illustration of the time slot and RB that have been defined in Fig.~\ref{fig:HelloProtocolModel}.
  The resource element (RE) is defined as the basic unit of a wireless resource, and it is assumed that one time slot consists a total of $B$ REs and a hello packet has a fixed length of $L$ bits.
  In order to convert the data rate unit into a more convenient one, consider $\tilde{\tau}=t\tau$ bits/RE for a positive constant $t$.
  Then, the number of REs for a hello packet transmission is given by $\textstyle N = \frac{L}{\tilde{\tau}}$.
\revbegin
  The size of one RB depends on $\tilde{\tau}$ or $\tau$, and the number of available RBs per time slot is given by\footnote{In practice, the parameters such as $B$, $N$, $M$, and $L$ are positive integers, but this paper relaxes the integer constraints for analytical convenience. That is, this paper allows that they are positive real numbers.}
\begin{IEEEeqnarray}{ll}\label{eq:Mxi}
  M & = \frac{B}{N} = \frac{B}{L/t\tau} = \frac{tB}{L} \log_2 (1+\frac{\xi}{\delta}) = \beta \log (1+\frac{\xi}{\delta}),
\end{IEEEeqnarray}
  where $\textstyle \beta \triangleq \frac{tB}{L\log2}$ is a constant.
  Hence, $\mathbb{E}\{S\}$ in \eqref{eq:EsNakaZn} is expressed as follows:
\begin{IEEEeqnarray}{ll}\label{eq:EsXi}
  \mathbb{E}\{S\} = \left(\frac{\Delta_s\left(m_s,\alpha \right)}{\Delta_i\left(m_i,\alpha \right)}(1-\rho)\beta\right) \frac{\log(1+\frac{\xi}{\delta})}{\xi^{\frac{2}{\alpha}}}.
\end{IEEEeqnarray}
  This $\mathbb{E}\{S\}$ can be maximized over $\xi>0$ by finding the optimal solution to maximize $\textstyle f_{\xi}(x) \triangleq x^{-\frac{2}{\alpha}} \log(1+\frac{x}{\delta})$.
  That is, $\xi^*=x^*$ where $\textstyle \xi^* \triangleq \arg \max_{\xi>0} \mathbb{E}\{S\}$ and $\textstyle x^* = \arg \max_{x>0} f_{\xi}(x)$.
\begin{proposition}\label{prop:Fxi}
  Function $f_\xi(x)$ on the domain of $\textstyle \{x|x>0\}$ is a unimodal function and has the maximum value at the unique solution of $u_\xi(x)=0$ for $x>0$ where $\textstyle u_\xi(x) \triangleq \frac{\alpha}{2}\frac{x}{\delta} - \left(1+\frac{x}{\delta}\right)\log(1+\frac{x}{\delta})$.
  Moreover, the optimal solution, i.e., $x^*$, increases with $\alpha$ for a fixed $\delta$.
\end{proposition}
\begin{proof}
  See Appendix~\ref{app:proof:prop:Fxi}.
\end{proof}
  The value of $\xi^*$ determined in Proposition~\ref{prop:Fxi} offers the optimal data rate and number of RBs for the broadcast of a hello packet, i.e., $\textstyle \tau^* = \log_2(1+\frac{\xi^*}{\delta})$ and $\textstyle M^*=\beta \log_2(1+\frac{\xi^*}{\delta})$.
  In Proposition~\ref{prop:Fxi}, $\xi^*$, $\delta$, and $\alpha$ are interestingly related.
  Let $\textstyle \tilde{u}_\xi(y) \triangleq \frac{\alpha}{2} y - \left( 1 + y \right) \log\left( 1 + y \right)$.
  Then, note that $\textstyle \tilde{u}_\xi(\frac{\xi^*}{\delta}) = 0$.
  This implies that $\textstyle \frac{\xi^*}{\delta}$ is only determined by $\alpha$.
  That is, if $\delta$ is scaled down by a factor of $s$, $\xi^*$ also decreases by the same factor, when $\alpha$ is given.
\revend


  Now, the focus returns to the optimal $\rho$.
  As mentioned before, the optimal $\rho$ can be decided regardless of the other parameters such as the wireless channel parameters, $M$, and $\xi$.
  It is interesting that $\mathbb{E}\{S\}$ given by \eqref{eq:EsNakaZn} increases linearly as $\rho$ decreases, and it implies that $\rho^* \rightarrow 0$.
  That is, by forcing $\rho$ to be extremely low, the network functions almost without packet collision as if each node transmits its hello packet through contention-free access.
  In this case, the maximum value of $\mathbb{E}\{S\}$ is upper bounded by and approaches $\textstyle 
  \frac{\Delta_s\left(m_s,\alpha \right)}{\Delta_i\left(m_i,\alpha \right)}  \frac{M}{\xi^{\frac{2}{\alpha}}}$.
  The assumption of zero noise power causes a typical node to ideally discover even far-off nodes.
  However, this is unrealistic when considering that each wireless link has a limited communication coverage due to $\sigma^2 > 0$.
  In Section~\ref{sec:NonzeroNoise}, the effect of a nonzero $\sigma^2$ will be considered.

\section{Spatial Analysis and Design Principles Considering a Nonzero Noise Power} \label{sec:NonzeroNoise}


  The nonzero noise power limits the communication range due to the finite SNR, and this limitation may result in different design principles to those of the interference limited scenario presented in Section~\ref{sec:ZeroNoise}.
  This section investigates the performance of wireless peer discovery when the effect of the noise cannot be ignored, and derives useful design principles by approximating the effect of the nonzero noise power.

  For mathematical tractability, the results in this section assume the specific values of the Nakagami-$m$ fading parameter and path loss exponent, i.e., $m_s=m_i=1$ and $\alpha=4$; these assumptions will be relaxed again in Section~\ref{sec:Results}, where the simulation results reveal that the design principles suggested in this section still works well even without these assumptions.


\subsection{Spatial Performance}\label{ssec:EsNonzeroNoise}

  The results for $\sigma^2=0$ derived in Section~\ref{sec:ZeroNoise} clearly reveal the inherent effect of the wireless channels and operation parameters; however, whether or not it is likely that the nonzero noise power changes their effects should be investigated.
  Fortunately, even when $\sigma^2>0$, if $m_s=m_i=1$ and $\alpha=4$, $\mathbb{E}\{S\}$ in \eqref{eq:EsNakaNzN} can be expressed in a simple form.
%
\begin{proposition}\label{prop:EsNzN}
  When all links experience Rayleigh fading, $\sigma^2>0$, and $\alpha=4$, $\mathbb{E}\{S\}$ is given by
\begin{IEEEeqnarray}{ll}\label{eq:EsRaylNzN}
  \hspace{-0.2cm} \mathbb{E}\{S\} &= \frac{\lambda \pi^{\frac{3}{2}} \rho(1-\rho)}{2\sqrt{\xi \sigma^2}} \exp\left(\left(\frac{\lambda \pi^2 \rho}{4M\sigma}\right)^2\right) \mathrm{erfc}\left(\frac{\lambda \pi^2 \rho}{4M\sigma}\right),
\end{IEEEeqnarray}
  where $\textstyle {\rm erfc}(x) = \frac{2}{\sqrt{\pi}}\int_{x}^{\infty} \exp\left(-t^2\right)dt$ is the complementary error function.
\end{proposition}
\begin{proof}
  By substituting $m_s=m_i=1$ into \eqref{eq:EsNakaNzN} and applying the integration formula of $\textstyle \int_0^\infty \exp \left( -(ax+bx^2) \right) dx = \frac{\sqrt{\pi}}{2\sqrt{b}} \exp \left( \frac{a^2}{4b} \right) \mathrm{erfc} \left( \frac{a}{2\sqrt{b}} \right) $ for $a \geq 0$ and $b>0$, \eqref{eq:EsRaylNzN} is derived.
\end{proof}
%
  It is interesting that $\mathbb{E}\{S\}$ in \eqref{eq:EsRaylNzN} depends on node density $\lambda$, unlike \eqref{eq:EsNakaZn} for $\sigma^2=0$.
  A finite transmit power limits the communication range; thus, a typical node cannot detect the signals of nodes outside this link coverage even when the aggregate interference power is low.
  Therefore, the number of nodes that exist within the link coverage affects the spatial performance of wireless peer discovery.
  It is worth noting that the coverage can be extended by increasing the transmit power, i.e., $p$ or $\textstyle \frac{p}{\tilde{\sigma}^2}=\frac{1}{\sigma^2}$.
  Accordingly, the design issue of $p$ arises, and this will be addressed at the end of the next subsection in detail.


\subsection{Design of the Discovery Operation Parameters}\label{ssec:DesignNonzeroNoise}

  In \eqref{eq:EsRaylNzN}, the discovery operation parameters of $M$, $\xi$, $\rho$, and $p$ (or $\textstyle \frac{1}{\sigma^2}$) are closely related to each other.
  Their joint design is optimal for maximizing $\mathbb{E}\{S\}$, but it is intractable.
  Therefore, this subsection elaborates on the impact of an individual parameter on $\mathbb{E}\{S\}$, and the joint optimization is left to a future work.


  Similar to the case of $\sigma^2=0$, when $\xi$ is fixed, it is clear that $\mathbb{E}\{S\}$ in \eqref{eq:EsRaylNzN} increases with $M$ because the aggregate interference decreases as $M$ increases.
  However, the difference with the result for $\sigma^2=0$ is that $\mathbb{E}\{S\}$ is a saturation function of $M$, i.e., $\textstyle \lim_{M\rightarrow\infty}\mathbb{E}\{S\} = \frac{\lambda \pi^{\frac{3}{2}} \rho(1-\rho)}{2\sqrt{\xi \sigma^2}}$ because $\lim_{x \rightarrow 0} \exp(x) = 1$ and $\lim_{x \rightarrow 0} \mathrm{erfc}(x) = 1$.
  That is, because the interference decreases as $M$ increases but the link coverage remains limited due to $\sigma^2 > 0$, $\mathbb{E}\{S\}$ is eventually saturated.

\revbegin
  It is difficult to derive the optimal value of $\xi$ for $\sigma^2>0$ because it should be designed jointly with $M$.
  For this reason, this paper does not mathematically derive the optimal $\xi$ for maximizing $\mathbb{E}\{S\}$ when $\sigma^2>0$; however, the numerical results in Section~\ref{sec:Results} demonstrate that it is likely that $\mathbb{E}\{S\}$ remains a unimodal function of $\xi$ even when $\sigma^2>0$.
  The analytical optimization of $\xi$ for $\sigma^2>0$ remains as future work.
\revend

  Moreover, $\mathbb{E}\{S\}$ is quite sensitive to $\rho$.
  It is not easy to derive $\rho^*$ directly from \eqref{eq:EsRaylNzN} in order to maximize $\mathbb{E}\{S\}$.
  However, the bounds of $\mathbb{E}\{S\}$ can be given in the form of a fractional function or linear function of $\rho$ from the bounds of $\mathrm{erfc}(x)$, i.e., $\textstyle \frac{1}{\sqrt{\pi}}\frac{2\tau}{1+2\tau^2}\exp(-\tau^2) < \mathrm{erfc}(\tau) < \frac{1}{\sqrt{\pi}} \frac{\exp(-\tau^2)}{\tau}$ for $\tau > 0$, as follows:
\begin{IEEEeqnarray}{l}\label{eq:EsRaylNzNBound}
  \frac{2}{\pi} \frac{M(1-\rho)}{\sqrt{\xi}} \frac{\kappa \rho^2}{1+\kappa \rho^2} < \mathbb{E}\{S\} < \frac{2}{\pi} \frac{M(1-\rho)}{\sqrt{\xi}},
\end{IEEEeqnarray}
  where $\textstyle \kappa \triangleq \frac{\lambda^2 \pi^4}{8 M^2 \sigma^2}$.
  Note that the upper bound in \eqref{eq:EsRaylNzNBound} is equal to \eqref{eq:EsNakaZnM} when $\alpha=4$.
  This implies that $\mathbb{E}\{S\}$ for $\sigma^2 > 0$ is upper bounded by that for $\sigma^2 = 0$.
  In addition, for a given $\rho > 0$, the lower bound in \eqref{eq:EsRaylNzNBound} becomes increasingly tight for a large $\kappa$, and it approaches the upper bound or that of \eqref{eq:EsNakaZnM} for $\alpha=4$ as $\sigma^2 \rightarrow 0$.
  Therefore, it can be understood that $\textstyle \frac{\kappa \rho^2}{1+\kappa \rho^2}$ of the lower bound in \eqref{eq:EsRaylNzNBound} simply and approximately models the performance degradation that results from the noise power, even if not exactly accurate.
  Unlike case of $\sigma^2=0$, it is expected that $\rho^*$ no longer approaches zero.
  A suboptimal $\rho$, i.e., $\hat{\rho}$, can be obtained in order to maximize the lower bound in \eqref{eq:EsRaylNzNBound} rather than $\rho^*$ for maximizing $\mathbb{E}\{S\}$.
  Let $\textstyle f_\rho(x) \triangleq \frac{x^2 (1-x)}{1+\kappa x^2}$ and $\textstyle \hat{x} \triangleq \arg \max_{0<x<1} f_\rho(x)$, then $\hat{\rho}=\hat{x}$.
%
\begin{proposition}\label{prop:Fx}
  Function $f_\rho(x)$ on the domain of $\textstyle \{x|0<x<1\}$ is a unimodal function and has a maximum value at the unique solution of $u_\rho(x)=0$ for $0<x<1$, where $u_\rho(x) = -\kappa x^3 - 3 x + 2$.
\end{proposition}
\begin{proof}
  The first order derivative of $f_\rho(x)$ with respect to $x$ is given by
\begin{IEEEeqnarray}{ll}\label{eq:PnDerivative}
  \frac{d f_\rho(x)}{dx} &= \frac{x}{(1+\kappa x^2)^2} u_\rho(x).
\end{IEEEeqnarray}
  For $0<x<1$, because $\textstyle \frac{x}{(1+\kappa x^2)^2} > 0$, the sign of $\textstyle \frac{d f_\rho(x)}{dx}$ is only determined by $u_\rho(x)$.
  Note that $u_\rho(0)=2$, $u_\rho(1)=-\kappa-1<0$, and $\textstyle \frac{du_\rho(x)}{dx} = -3\kappa x^2 - 3 < 0$.
  That is, because $u_\rho(x)$ is monotonically decreasing, $u_\rho(x)>0$ for $0<x<\hat{x}$ while $u_\rho(x)<0$ for $\hat{x} < x<1$.
  Also, $u_\rho(\hat{x})=0$.
  Therefore, $f_\rho(x)$ is monotonically increasing for $0<x \leq \hat{x}$ and monotonically decreasing for $\hat{x} \leq x<1$, and $\hat{x}$ becomes the optimal solution to maximize $f_\rho(x)$ on the interval of $0<x<1$.
\end{proof}
  Interestingly, this $\hat{\rho}$ is closely related to the environmental factors including $\lambda$, $\sigma^2$, and $M$, because it depends on $\kappa$, and this design of $\rho$ differs significantly from that for $\sigma^2=0$ in Section~\ref{ssec:DesignNakaZn}.

  Until now, the design of the three parameters, i.e., $M$, $\xi$, and $\rho$, was addressed for a given $\textstyle \sigma^2=\frac{\tilde{\sigma}^2}{p}$.
  As another design method, it can be considered that $p$ is set to such a large value that the aggregate interference dominates the noise power.
  This work suggests a design method for a transmit power that can suppress the effect of the noise power and be kept as small as possible, similar to \cite{Kwon13_SG_RandomDataCollector}.
  The key is that, in \eqref{eq:EsRaylNzNBound}, the lower bound is forced to approach the upper bound by designing $p$ that makes $\kappa\rho^2 \gg 1$.
  From $\textstyle \kappa \rho^2 = \frac{\lambda^2 \pi^4 \rho^2}{8 M^2 \sigma^2} = c \gg 1$ for a certain large value of $c$, $p$ can be set to $\textstyle \hat{p} \triangleq \frac{8c}{\pi^4}\left(\frac{\lambda \rho}{M}\right)^{-2} \tilde{\sigma}^2$. 
  It is worth noting that $\hat{p}$ is a decreasing function of the interferer density, i.e., $\textstyle \frac{\lambda \rho}{M}$.
  That is, in this design, the noise power is dominated by the aggregate interference power.
  $\mathbb{E}\{S\}$ using this $\hat{p}$ approaches that for $\sigma^2=0$, which eventually facilitates the application of the design principles for $M$, $\xi$, and $\rho$ addressed in Section~\ref{ssec:DesignNakaZn}.

\revbegin

\section{The Impact of Shadowing on Wireless Peer Discovery} \label{sec:Shadowing}

  This section discusses the impact of the shadowing on the spatial performance of wireless peer discovery and extends the design principles derived in previous sections to the ones for wireless channel models that incorporate arbitrary shadowing.

  When considering the link from node $j$ to a typical node under shadowing and Nakagami-$m$ fading modeled by $\vartheta_j$ and $\tilde{h}_j$, respectively, the received power can be written using $\textstyle p \tilde{h}_j \vartheta_j |X_j|^{-\alpha} = p \tilde{h}_j |\vartheta_j^{-\frac{1}{\alpha}} X_j|^{-\alpha}$.
  In this statement, $\textstyle \vartheta_j^{-\frac{1}{\alpha}} X_j$ can be interpreted as randomly and independently displacing nodes of $\Phi$ to a new location according to their shadowing \cite{Dhillon13_DlHetnetGen}.
  This paper assumes that $\{\vartheta_j\}$ are i.i.d.
  With a slight misuse of notation, $\vartheta$ denotes a random variable representing the shadowing component of the links.
  From Lemma~1 in \cite{Dhillon13_DlHetnetGen}, which follows from the displacement theorem \cite[Theorem~1.3.9]{Baccelli10_StochasticGeomNow}, if $\mathbb{E}\{\vartheta^{\frac{2}{\alpha}}\} < \infty$, the i.i.d. shadowing effect is equivalent to the transformation of original PPP $\Phi$ with density $\lambda$ into new homogeneous PPP $\textstyle \Phi^{(\vartheta)}$ with density $\textstyle \lambda^{(\vartheta)}$, which is given by
\begin{IEEEeqnarray}{ll}\label{eq:DensityShadowing}
  \lambda^{(\vartheta)} \triangleq \lambda \mathbb{E}\{\vartheta^{\frac{2}{\alpha}}\}.
\end{IEEEeqnarray}

  This concept enables the investigation of the shadowing effect through only analyzing the effect of $\lambda$ on the performance.
  This results in the following interpretation under an interference-limited scenario, i.e., when $\sigma^2=0$.
\begin{remark}\label{rmk:ShadowingZeroNoise}
  When $\sigma^2=0$, the average number of successfully discovered nodes, i.e., $\mathbb{E}\{S\}$, is invariant to the shadowing distribution.
\end{remark}
\begin{proof}
  $\mathbb{E}\{S\}$ given in \eqref{eq:EsNakaZn} is independent of $\lambda$.
  Therefore, the shadowing effect described by \eqref{eq:DensityShadowing} does not affect the performance.
\end{proof}
  Remark~\ref{rmk:ShadowingZeroNoise} indicates that the design principles for $M$, $\rho$, and $\xi$ described in Section~\ref{sec:ZeroNoise} do not depend on the shadowing distribution, when $\sigma^2=0$.
  In contrast, when the noise power cannot be neglected, the shadowing impact on the performance is revealed from the following results.
\begin{remark}\label{rmk:ShadowingNonzeroNoise}
  When all links experience the Rayleigh fading, $\sigma^2>0$, and $\alpha=4$, $\mathbb{E}\{S\}$ increases with $\mathbb{E}\{\vartheta^{\frac{2}{\alpha}}\}$.
  In addition, as $\mathbb{E}\{\vartheta^{\frac{2}{\alpha}}\}$ increases, $\mathbb{E}\{S\}$ approaches that of the interference limited case, i.e., $\textstyle \frac{2}{\pi}\frac{M(1-\rho)}{\sqrt{\xi}}$, which is an upper bound of $\mathbb{E}\{S\}$ for $\sigma^2>0$                                    .
\end{remark}
\begin{proof}
  The increase in $\mathbb{E}\{S\}$ with $\mathbb{E}\{\vartheta^{\frac{2}{\alpha}}\}$ can be demonstrated through proving that $\textstyle \frac{\partial \mathbb{E}\{S\}}{\partial \lambda} > 0$, where $\mathbb{E}\{S\}$ is given in \eqref{eq:EsRaylNzN}.
\begin{IEEEeqnarray}{ll}\label{eq:SavgDiffLambda}
  \hspace{-0.4cm} \frac{\partial \mathbb{E}\{S\}}{\partial \lambda} & = c \left( 1 + 2d^2 \lambda^2 \right) \exp \left( d^2\lambda^2 \right) \mathrm{erfc}\left( d\lambda \right) - \frac{2cd}{\sqrt{\pi}}\lambda \IEEEnonumber \\
  & \stackrel{\mathrm{(a)}}{>}  c \left( 1 + 2d^2 \lambda^2 \right) \frac{1}{\sqrt{\pi}} \frac{2d\lambda}{1 + 2d^2\lambda^2} - \frac{2cd}{\sqrt{\pi}}\lambda = 0,
\end{IEEEeqnarray}
  where $\textstyle c \triangleq \frac{\pi^{\frac{3}{2}} \rho(1-\rho)}{2\sqrt{\xi\sigma^2}}$, $\textstyle d \triangleq \frac{\pi^2 \rho}{4 M \sigma}$,
    and (a) follows from $\mathrm{erfc}(\tau) > \textstyle \frac{1}{\sqrt{\pi}}\frac{2\tau}{1+2\tau^2}\exp(-\tau^2)$. 
  As $\lambda$ increases, the lower bound of $\mathbb{E}\{S\}$ in \eqref{eq:EsRaylNzNBound} approaches the upper bound in \eqref{eq:EsRaylNzNBound}, which is $\mathbb{E}\{S\}$ for $\sigma^2=0$; thus, $\mathbb{E}\{S\}$ approaches $\textstyle \frac{2}{\pi}\frac{M(1-\rho)}{\sqrt{\xi}}$ as $\mathbb{E}\{\vartheta^{\frac{2}{\alpha}}\}$ increases.
\end{proof}
  When $\mathbb{E}\{\vartheta^{\frac{2}{\alpha}}\} > 1$, effective density $\lambda^{(\vartheta)}$ in \eqref{eq:DensityShadowing} is larger than original density $\lambda$, and this results in reducing the effect of the nonzero noise power due to the increase in the effective node density.
  That is, shadowing may cause the operation of wireless peer discovery to be closer to that of an interference-limited scenario.
  By replacing $\lambda$ with $\lambda^{(\vartheta)}$, the design principles in Section~\ref{ssec:DesignNonzeroNoise} also work well under shadowing.

  Recall that Remarks~\ref{rmk:ShadowingZeroNoise} and \ref{rmk:ShadowingNonzeroNoise} are applicable to an arbitrary distribution of shadowing component $\vartheta$.
  However, the shadowing component is most commonly modeled as lognormal, i.e., $\theta$ such that $\theta \triangleq 10^{\frac{\vartheta}{10}}$ can be represented as a normal random variable with a zero mean and standard deviation $\chi$.
  In this case, $\textstyle \mathbb{E}\{\vartheta^{\frac{2}{\alpha}}\} = \exp\left( \frac{1}{2}\left( \frac{\log 10}{5} \frac{\chi}{\alpha} \right)^2 \right)$.
  That is, $\mathbb{E}\{\vartheta^{\frac{2}{\alpha}}\}$ increases with $\chi$, and $\mathbb{E}\{\vartheta^{\frac{2}{\alpha}}\} > 1$ because $\chi > 0$.
  From Remark~\ref{rmk:ShadowingNonzeroNoise}, this signifies that $\mathbb{E}\{S\}$ increases with $\chi$ and the lognormal shadowing always leads to $\mathbb{E}\{S\}$ larger than that without shadowing when the effect of the noise power cannot be ignored.
  Recall that Remark~\ref{rmk:ShadowingNonzeroNoise} assumes a specific channel model, i.e., the Rayleigh fading and $\alpha=4$.
  In the next section, simulation results demonstrate that this property remains under other channel models, i.e., the Nakagami-$m$ fading and general $\alpha$.
\revend

\section{Numerical Results} \label{sec:Results}


  This section evaluates and discusses the spatial performance of a multichannel random hello protocol based on the results derived in Sections~\ref{sec:ZeroNoise} and \ref{sec:NonzeroNoise}.
  As mentioned in Section~\ref{ssec:Model}, $\textstyle \frac{p}{\tilde{\sigma}^2}=\frac{1}{\sigma^2}$ is the average received SNR when a target node is at a unit distance from a typical node, i.e., $|X_i|=1$ in \eqref{eq:Sinr}; hereafter, \texttt{snr} signifies the average received SNR at the unit distance.
  It is assumed that nodes are spatially distributed according to a homogeneous PPP.
  Node density $\lambda$ is measured as the average number of nodes within a unit area and is set to $4$, if not stated otherwise.
  For this value, the average distance between nodes is $\textstyle 1/(2\sqrt{\lambda})=0.25$ \cite{Baccelli10_StochasticGeomNow}.

\begin{figure}
\centering
\includegraphics[width=9cm]{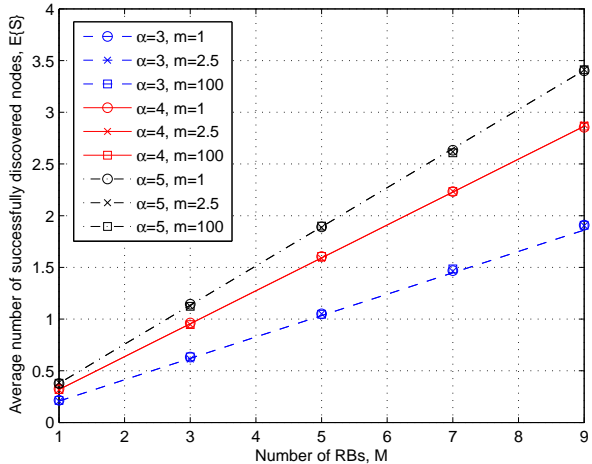}
\vspace{-0.5cm}
\caption{$\mathbb{E}\{S\}$ vs. $M$ when $\sigma^2=0$ ($\xi=0$dB, $\rho=0.5$; lines: analysis results, symbols: simulation results).} \label{fig:ResultSvsMzN}
\end{figure}
\begin{figure}
\centering
\includegraphics[width=9cm]{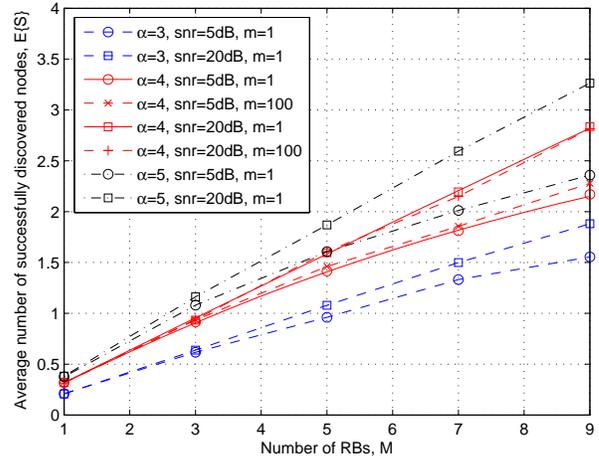}
\vspace{-0.5cm}
\caption{$\mathbb{E}\{S\}$ vs. $M$ when $\sigma^2>0$ ($\xi=0$dB, $\rho=0.5$; note that, in the cases with $\alpha=4$ \& $m=1$, lines and symbols denote the analysis and simulation results, respectively. In the other cases, symbolled lines only represent the simulation results.).} \label{fig:ResultSvsMNnzN}
\end{figure}
%


  Fig.~\ref{fig:ResultSvsMzN} presents the effects of the number of RBs $M$, wireless fading channel parameterized by $m$, and path loss exponent $\alpha$ on $\mathbb{E}\{S\}$ under the Nakagami-$m$ fading channel models with $m_s=m_i=m$, when $\sigma^2=0$.
  As already expected, $\mathbb{E}\{S\}$ increases with $M$ and $\alpha$, and it does not depend on the value of the fading parameter $m$.
  This figure also depicts that the analysis results coincide precisely with the simulation results.
  The result in Proposition~\ref{prop:EsNakaZn} only provides the result for integer $m$ while Fig.~\ref{fig:ResultSvsMzN} demonstrates that $\mathbb{E}\{S\}$ also remains independent of $m$ with a non-integer value.

  In order to observe the effect of the noise power, {\rev the finite \texttt{snr}} is considered in Fig.~\ref{fig:ResultSvsMNnzN}.
  This figure reveals clearly that $\mathbb{E}\{S\}$ tends to be saturated rather than continuously increasing as $M$ increases when \texttt{snr} is low, i.e. $5$dB.
  For example, when $\alpha=4$, $\sigma^2=-5$dB, and $m=1$, $\mathbb{E}\{S\}$ eventually approaches $\textstyle \frac{\lambda \pi^{\frac{3}{2}} \rho(1-\rho)}{2\sqrt{\xi \sigma^2}} = 4.3503$ when $M$ is $1000$ and $\rho$ is $0.5$.
  The simulation results in Fig.~\ref{fig:ResultSvsMNnzN} also demonstrate the effect of $m$ and $\alpha$ with values other than $m=1$ and $\alpha=4$ that are assumed in Section~\ref{sec:NonzeroNoise}.
  In this figure, it is observed that the effect of $m$ remains insignificant and $\mathbb{E}\{S\}$ increases with $\alpha$, which is similar to the interference limited case.

%
\begin{figure}[t]
\centering
\includegraphics[width=9cm]{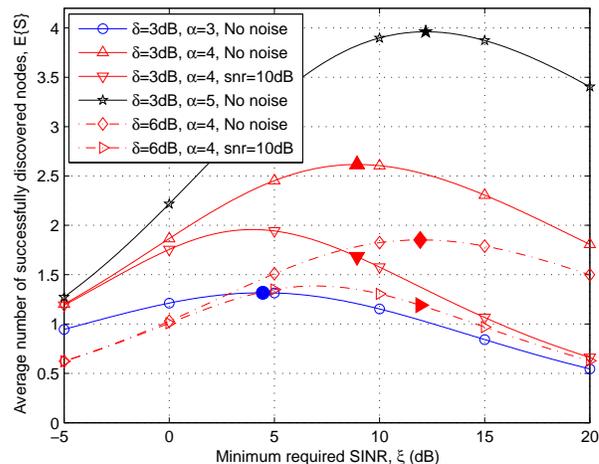}
\vspace{-0.5cm}
\caption{$\mathbb{E}\{S\}$ vs. $\xi$ under the Rayleigh fading environment ($\rho=0.5$; closed symbols represent $\xi^*$ derived in Proposition~\ref{prop:Fxi}).} \label{fig:ResultSvsXidB}
\end{figure}

\revbegin
  Fig.~\ref{fig:ResultSvsXidB} presents the performance gains that the design of $\xi$ suggested in Section~\ref{ssec:DesignNakaZn} enables.
  In order to determine the value of $\xi$, $\beta$ defined as $\textstyle \frac{tB}{L\log2}$ in Section~\ref{ssec:DesignNakaZn} should be chosen appropriately.
  When considering the uplink resource structure of the 3GPP LTE system \cite{LTE_PhyChannel,LTE_UeTxRx} and assuming that $L=70\,$bits \cite{Bacelli12_SG_D2Ddiscovery}, $\beta$ ranges from approximately $10$ to $160$ depending on the available bandwidth, and therefore the evaluation assumes that $\beta=10$.
  These results demonstrate that $\mathbb{E}\{S\}$ for $\sigma^2=0$ is maximized at the value of $\xi$ derived in Proposition~\ref{prop:Fxi}, which increases with $\alpha$ for a given $\delta$.
  This figure also depicts that $\xi^*$ increases with the SNR gap under a fixed $\alpha$.
  As stated in Section~\ref{ssec:DesignNakaZn}, it should be noted that, when $\sigma^2 = 0$ and $\alpha = 4$, $\xi^*$'s for $\delta = 6\,$dB and $\delta = 3\,$dB differ precisely by the difference of $\delta$ values, i.e., $3\,$dB.
  In contrast, it is observed that $\xi$ required in order to maximize $\mathbb{E}\{S\}$ for a finite \texttt{snr} is somewhat smaller than $\xi^*$ derived in Proposition~\ref{prop:Fxi}.
  Even though $\xi^*$ derived in Proposition~\ref{prop:Fxi} does not provide a bad performance, it might not be satisfactory to apply this $\xi^*$ when the SNR is not high.
  However, as suggested in Section~\ref{ssec:DesignNonzeroNoise}, if the transmit power can be appropriately increased, it is expected that this $\xi^*$ can work sufficiently well.
\revend

\begin{figure}[t]
\centering
\includegraphics[width=9cm]{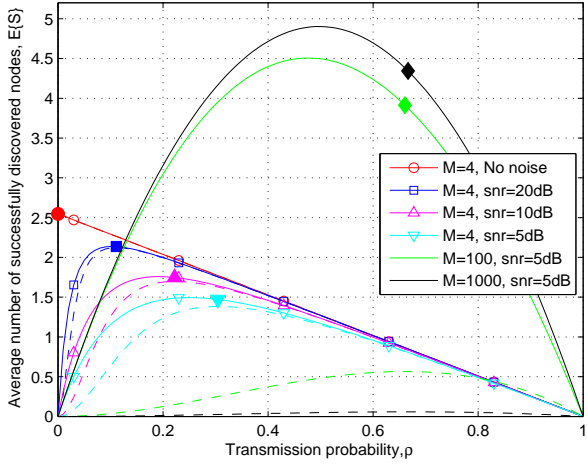}
\vspace{-0.5cm}
\caption{$\mathbb{E}\{S\}$ vs. $\rho$ under the Rayleigh fading environment ($\alpha=4$, $\xi=0$dB; open symbolled solid lines: exact values of $\mathbb{E}\{S\}$, dashed lines: lower bound of $\mathbb{E}\{S\}$ in \eqref{eq:EsRaylNzNBound}, closed symbols: $\hat{\rho}$ derived in Proposition~\ref{prop:Fx}).} \label{fig:ResultSvsRho}
\end{figure}

  Fig.~\ref{fig:ResultSvsRho} demonstrates how the transmission probability, i.e., $\rho$, affects $\mathbb{E}\{S\}$ and how well the suboptimal design of $\rho$ proposed in this paper functions under the Rayleigh fading environment ($m_s=m_i=1$).
  It is observed that the optimal $\rho$ for maximizing $\mathbb{E}\{S\}$, i.e., $\rho^*$, increases with $\sigma^2$.
  This increases the likelihood of packets with a high received SNR by allowing more nodes to transmit rather than only focusing on reducing interference.
  When $\sigma^2$ is low, $\mathbb{E}\{S\}$ is very sensitive to $\rho$ on the interval of $0<\rho<\rho^*$; thus, the selection of $\rho$ has a profound effect on the performance.
  This observation stresses the importance of considering the noise power effect in the design of $\rho$.
  The results in Fig.~\ref{fig:ResultSvsRho} also demonstrate that {\rev $\hat{\rho}$ obtained in Proposition~\ref{prop:Fx} tracks $\rho^*$ very well even when the lower bound of $\mathbb{E}\{S\}$ in \eqref{eq:EsRaylNzNBound} becomes increasingly loosed as the impact of the noise power increases, e.g., for a low SNR or large $M$.}
  That is, $\hat{\rho}$  achieves a fairly good balance between the chance of packet transmission and reduction of interference.
  Fig.~\ref{fig:ResultSvsRho} also illustrates that $\rho^*$ approaches $0.5$ for a large $M$.
  In fact, $\rho^*=0.5$ maximizes $\textstyle \lim_{M\rightarrow\infty}\mathbb{E}\{S\}$. 
  At a large $M$, $\hat{\rho}$ tends to be more than $0.5$.
  This value provides a lower performance than that of a trivial design of $\rho=0.5$ and wastes the node energy due to more transmissions.
  From this, a design method may be considered where $\rho$ is set to $\min\{\hat{\rho},0.5\}$.
  The results also present that a value of $\rho$ ranging from $0.1$ to $0.3$ works moderately well over various \texttt{snr} values.
  Accordingly, in what follows, $\rho=0.2$ will be used for performance evaluations if not stated otherwise.

\begin{figure}[t]
\centering
\includegraphics[width=9cm]{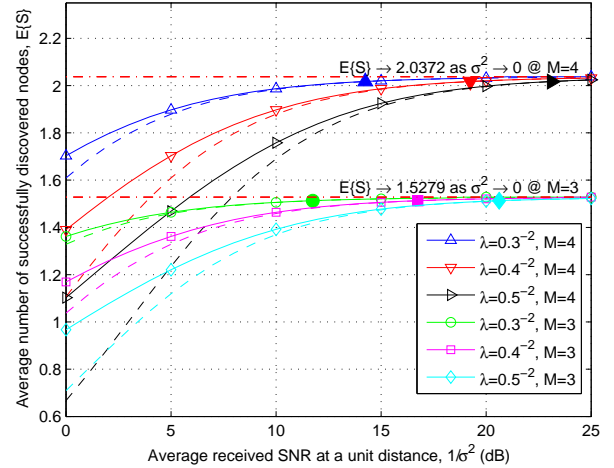}
\vspace{-0.5cm}
\caption{$\mathbb{E}\{S\}$ vs. \texttt{snr} under the Rayleigh environment ($\alpha=4$, $\xi=0$dB, $\rho=0.2$; open-symbolled solid lines: exact values of $\mathbb{E}\{S\}$, dashed \& dashdot lines: lower \& upper bounds of $\mathbb{E}\{S\}$ in \eqref{eq:EsRaylNzNBound}, closed symbols: \texttt{snr} designed for meeting $\kappa\rho^2=100$).} \label{fig:ResultSvsSnrdB}
\end{figure}

  Fig.~\ref{fig:ResultSvsSnrdB} elaborates on the approximation of the noise power impact based on the lower bound in \eqref{eq:EsRaylNzNBound}, and it validates the design of the transmit power suggested in Section~\ref{ssec:DesignNonzeroNoise}.
  It is observed that the approximation becomes more and more precise as $\textstyle \kappa=\frac{\lambda^2 \pi^4}{8 M^2 \sigma^2}$ increases.
  This is because the noise power is overwhelmed by the aggregate interference that increases with $\lambda/M$ and \texttt{snr}.
  In this figure, the solid symbols denote the performance of the transmit power design suggested in Section~\ref{ssec:DesignNonzeroNoise}, when $c=\kappa\rho^2=100$, and a transmit power is expressed as the average received SNR at a unit distance in the abscissa.
  The results demonstrate that this design gives a good transmit power that can be maintained as small as possible while forcing it into an interference limited environment.


%
\begin{figure}[ht]
\centering
\includegraphics[width=9cm]{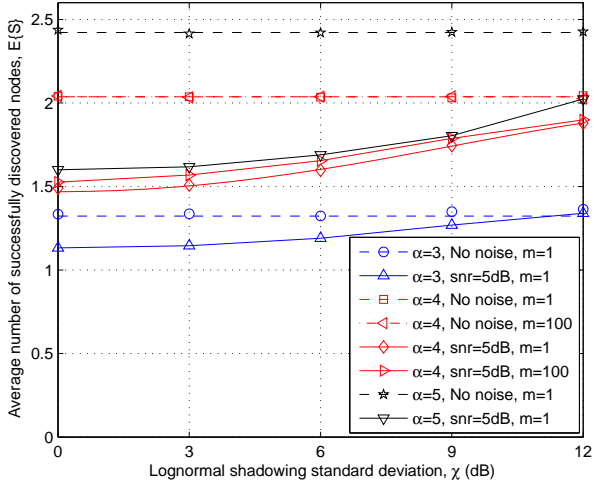}
\vspace{-0.5cm}
\caption{{\rev The effect of the wireless channels on $\mathbb{E}\{S\}$ when considering the lognormal shadowing ($M=4$, $\xi=0$dB, $\rho=0.2$; note that, in the cases of `No noise' and the case of `$\alpha=4$, \texttt{snr}$=5$dB, $m=1$', lines and symbols denote the analysis and simulation results, respectively. In the other cases, symbolled lines only represent the simulation results).}}
\label{fig:ResultSvsLNdBAnaSim}
\end{figure}

\revbegin
  Fig.~\ref{fig:ResultSvsLNdBAnaSim} presents the effect of the lognormal shadowing on $\mathbb{E}\{S\}$, and the value of $\chi$ on the abscissa denotes the standard deviation of the shadowing in dB scale, e.g., $\chi=0$ indicates no shadowing.
  These results verify the discussion presented in Section~\ref{sec:Shadowing} by demonstrating the coincidence of the analysis and simulation results under the interference limited (i.e., $\sigma^2=0$ and general $\alpha$) and specific nonzero noise (i.e., $\sigma^2>0$, $\alpha=4$, and $m=1$) channels while presenting the simulation results under analytically intractable channels (e.g., $\sigma^2>0$ and $\alpha \neq 4$).
  As elaborated in Section~\ref{sec:Shadowing}, the performance is invariant to $\chi$ when $\sigma^2=0$ while it increases with $\chi$ when $\sigma^2>0$.
  In particular, when $\sigma^2>0$, the lognormal shadowing equivalently increases the node density by a factor of $\textstyle \exp\left( \frac{1}{2}\left( \frac{\log 10}{5} \frac{\chi}{\alpha} \right)^2 \right)$.
  In this regard, the shadowing has more significant impact on the performance for a small $\alpha$, and this phenomenon is observed in Fig.~\ref{fig:ResultSvsLNdBAnaSim}.

  Fig.~\ref{fig:ResultDesignSh} presents the effect of shadowing on the design of the operation parameters.
  The two subfigures extend the results in Figs.~\ref{fig:ResultSvsXidB} and \ref{fig:ResultSvsRho} into the ones that incorporate the lognormal shadowing.
  The lognormal shadowing tends to dilute the impact of the noise power; thus, it is observed that $\xi$ and $\rho$ for maximizing $\mathbb{E}\{S\}$ become closer and closer to those of the interference limited case, as $\chi$ increases.
  In addition, Fig.~\ref{fig:ResultSvsRhoSh} demonstrates that the maximum value of $\mathbb{E}\{S\}$ at $\chi=12$ increases by up to $18$\% when \texttt{snr}$=10$dB while increasing by only $8$\% when \texttt{snr}$=20$dB, compared with the case with not shadowing.
  This observation implies that the impact of shadowing reduces, as $\sigma^2$ decreases or \texttt{snr} increases.
\revend

\section{Conclusions} \label{sec:Conclusions}

%
\begin{figure*}[ht]
\centerline{
    \subfigure[{\rev Design of $\xi$ ($\rho = 0.5$, $\beta=10$).}]{\includegraphics[width=9cm]{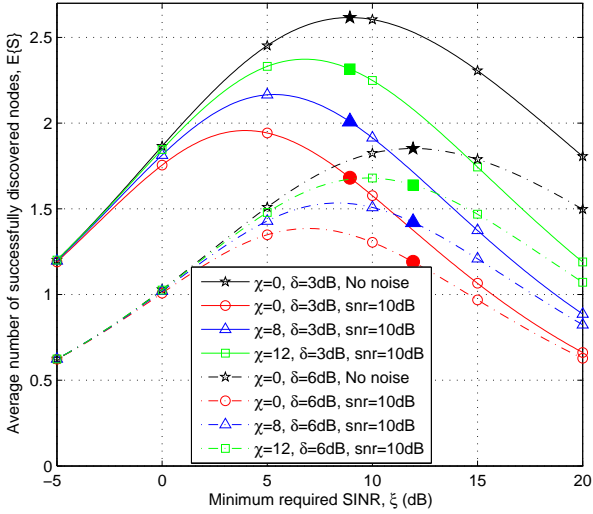}
    \label{fig:ResultSvsXidBSh}}
\hfil
    \subfigure[{\rev Design of $\rho$ ($\xi = 0$dB).}]{\includegraphics[width=9cm]{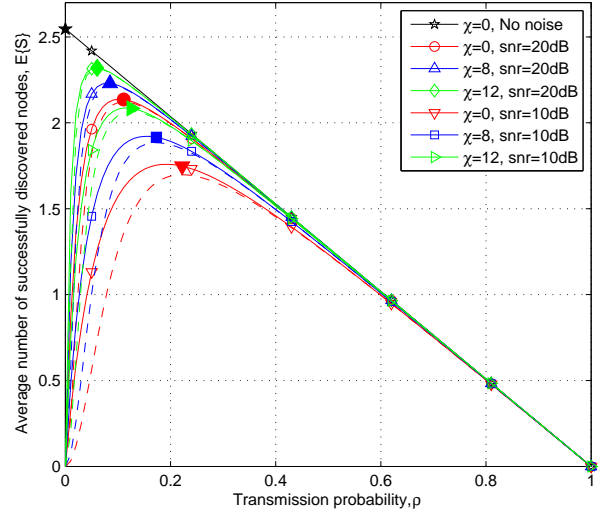}
    \label{fig:ResultSvsRhoSh}}}
    \caption{{\rev The impact of lognormal shadowing on operation parameter design ($M=4$, $m=1$; (a) closed symbols: $\xi^*$ derived in Proposition~\ref{prop:Fxi}; (b) closed symbols: $\hat{\rho}$ derived in Proposition~\ref{prop:Fx}, dashed lines: lower bound of $\mathbb{E}\{S\}$ in \eqref{eq:EsRaylNzNBound}).}}
\label{fig:ResultDesignSh}
\vspace{0cm}
\end{figure*}
%


  This paper investigated the performance of a multichannel random hello protocol for wireless peer discovery in terms of the average number of successfully discovered peers, when the peers are spatially distributed according to a homogeneous Poisson point process.
  The performance depends on the wireless channel characteristics, such as the path loss, noise power, fading, and shadowing, as well as the discovery operation characteristics, such as the number of resource blocks, modulation and coding scheme, transmission probability, and transmit power.
  The relationship among these characteristics was expressed or approximated as a closed form, and it was demonstrated that the wireless channel model significantly affects the design of the discovery operation parameters.
  Accordingly, incorrect models or assumptions might result in poor designs, e.g., it was observed that an immoderate zero noise assumption for a low SNR might lead a poor design of the transmission probability that degrades the performance.
  The results in this paper can be used as a basis for the design of wireless peer discovery, even though this paper only considered a limited scenario including a simple random access and homogeneous PPP.
  For future study, it would be interesting to extend this work by considering more sophisticated resource management schemes, e.g., interference aware resource allocation.

\appendices

\section{Proof of Lemma~\ref{lem:SnakaNzN}} \label{app:proof:lem:SnakaNzN}

%
\begin{IEEEeqnarray}{ll}
  \hspace{-0.4cm} \mathbb{E}\{S\} & \stackrel{\mathrm{(a)}}{=} M \mathbb{E}\left\{ (1-\rho) \sum_{X_i \in \Phi_{1/M}} P(X_i|Z_0=0) \right\} \nonumber \\
  & \stackrel{\mathrm{(b)}}{=} M \left( (1-\rho)\frac{\lambda}{M} \int_{X\in\mathbb{R}^2} P(X|Z_0=0) dX \right) \nonumber \\
  & \stackrel{\mathrm{(c)}}{=} 2\pi \lambda \rho (1-\rho) \int_0^\infty \Pr \left\{ h > \xi r^\alpha (I + \sigma^2) \right\} r dr, \label{eq:EsGenDrv}
\end{IEEEeqnarray}
  where (a) follows from the fact that a typical node simultaneously listens to target nodes over $M$ RBs, (b) follows from the Campbell theorem and the stationarity of a homogeneous PPP \cite{Stoyan96_StochasticGeom}, and (c) follows from \eqref{eq:Sinr}, \eqref{eq:PxDef}, and the change of variable $|X| \rightarrow r$.
  Here, $\mathbb{R}^2$ denotes the two-dimensional Euclidean space.
  When the fading power gains of the desired and interfering signals have the ccdfs with fading parameters $m_s$ and $m_i$ described by \eqref{eq:NakaCcdf}, respectively, $\textstyle \Pr \left\{ h > \xi r^\alpha (I + \sigma^2) \right\}$ can be derived similarly to equations (20), (21), and (24) in \cite{Kwon13_SG_RandomDataCollector}, as follows:
\begin{IEEEeqnarray}{l}\label{eq:SinrCcdfNakaDrv}
  \Pr \left\{ h > \xi r^\alpha (I + \sigma^2) \right\} = \sum_{k=0}^{m_s-1} \frac{m_s^k}{k!}(-\xi r^\alpha)^k \cdot \IEEEnonumber \\
    \hspace{2.5cm} \left.\frac{d^k \mathcal{L}_I(\zeta) \exp\left(-\zeta \sigma^2\right)}{d \zeta^k}\right|_{\zeta=m_s \xi r^\alpha},
\end{IEEEeqnarray}
  where $\textstyle \mathcal{L}_I(\zeta) = \exp\left( -\frac{\lambda\rho}{M} \pi\zeta^{\frac{2}{\alpha}} \Delta_i (m_i,\alpha) \right)$.
  Eventually, by plugging \eqref{eq:SinrCcdfNakaDrv} into \eqref{eq:EsGenDrv}, \eqref{eq:EsNakaNzN} is obtained.
\hfill\IEEEQED

\section{Proof of Proposition~\ref{prop:EsNakaZn}} \label{app:proof:prop:EsNakaZn}

  By using the formula of $\textstyle \frac{\partial^k}{\partial z^k} \exp(f(z)) = \exp(f(z)) \sum_{l=0}^{k} \frac{1}{l!} \sum_{j=0}^{l} (-1)^j \binom{l}{j} f(z)^j \frac{\partial^k f(z)^{l-j}}{\partial z^k}$ similar to equation (25) in \cite{Kwon13_SG_RandomDataCollector}, the higher order derivative terms in \eqref{eq:EsNakaNzN} are given by
\begin{IEEEeqnarray}{ll}\label{eq:KthDer}
  \frac{d^k \exp\left( -\frac{\lambda\rho}{M} \phi(\zeta,\alpha) \right)}{d\zeta^k} = \exp\left( -\frac{\lambda \rho}{M}\pi\zeta^{\frac{2}{\alpha}}\Delta_i(m_i,\alpha) \right) \cdot \IEEEnonumber \\
  \hspace{1.5cm}\sum_{l=0}^k \frac{1}{l!} \sum_{j=0}^l (-1)^{l+j} \binom{l}{j} \left( \frac{\lambda \rho}{M}\pi\Delta_i(m_i,\alpha) \right)^l \cdot \IEEEnonumber \\
    \hspace{2.8cm} \left(\frac{2}{\alpha}(l-j)\right)_{(k)} \zeta^{\frac{2}{\alpha}l-k},
\end{IEEEeqnarray}
  where $(x)_{(k)} \triangleq x(x-1)\cdots (x-k+1)$ denotes the Pochhammer symbol.
  Thus, through the integration similar to equation (26) in \cite{Kwon13_SG_RandomDataCollector}, \eqref{eq:EsNakaNzN} is calculated as follows.
\begin{IEEEeqnarray}{ll}\label{eq:EsMint}
  \mathbb{E}\{S\} = \frac{\tilde{\Delta}_s(m_s,\alpha)}{\Delta_i(m_i,\alpha)} \frac{M(1-\rho)}{\xi^{\frac{2}{\alpha}}},
\end{IEEEeqnarray}
  where $\textstyle \tilde{\Delta}_s(m_s,\alpha)$ denotes $\textstyle m_s^{-\frac{2}{\alpha}} \sum_{k=0}^{m_s-1} \frac{1}{k!} \sum_{l=0}^k \sum_{j=0}^{l}\allowbreak (-1)^{k+l+j} \binom{l}{j} \left(\frac{2}{\alpha}(l-j)\right)_{(k)}$.

  Interestingly, a more elegant form of $\textstyle \tilde{\Delta}_s(m_s,\alpha)$ can be found.
  That is, it is shown that $\tilde{\Delta}_s(m_s,\alpha) = \textstyle \Delta_s (m_s,\alpha)$, as follows:
%
\begin{IEEEeqnarray}{ll}\label{eq:DeltaS}
  \tilde{\Delta}_s & (m_s,\alpha) \nonumber \\
  & \stackrel{\mathrm{(a)}}{=} m_s^{-\frac{2}{\alpha}} \sum_{k=0}^{m_s-1} \frac{1}{k!} \sum_{n=0}^k (-1)^{k+n} \left(\frac{2}{\alpha}n\right)_{(k)} \sum_{l=n}^{k} \binom{l}{n} \nonumber \\
  & \stackrel{\mathrm{(b)}}{=} m_s^{-\frac{2}{\alpha}} \sum_{k=0}^{m_s-1} \frac{1}{k!} \sum_{n=0}^k (-1)^{k+n} \binom{k+1}{n+1} \left(\frac{2}{\alpha}n\right)_{(k)} \nonumber \\
  & \stackrel{\mathrm{(c)}}{=} m_s^{-\frac{2}{\alpha}} \sum_{k=0}^{m_s-1} \frac{1}{k!} \sum_{s=0}^{t-1} (-1)^{s} \binom{t}{s} \left(\frac{2}{\alpha}(t-s)-\frac{2}{\alpha}\right)_{(k)} \nonumber \\
  & = m_s^{-\frac{2}{\alpha}} \sum_{k=0}^{m_s-1} \frac{1}{k!} \left(\left(\sum_{s=0}^{t} (-1)^{s} \binom{t}{s} \left(\frac{2}{\alpha}(t-s)-\frac{2}{\alpha}\right)_{(k)} \right) \right. \nonumber \\
  &\hspace{1.5cm} \left. - \left( (-1)^{k+1} \left(-\frac{2}{\alpha}\right)_{(k)}  \right) \right) \nonumber \\
  & \stackrel{\mathrm{(d)}}{=} m_s^{-\frac{2}{\alpha}} \sum_{k=0}^{m_s-1} \frac{1}{k!} \frac{\Gamma\left(\frac{2}{\alpha}+k\right)}{\Gamma\left(\frac{2}{\alpha}\right)}  \nonumber \\
  & = m_s^{-\frac{2}{\alpha}} \frac{1}{\Gamma\left(\frac{2}{\alpha}\right)} \sum_{k=0}^{m_s-1} \frac{\Gamma\left(k+\frac{2}{\alpha}\right)}{\Gamma\left(k+1\right)} \nonumber \\
  & \stackrel{\mathrm{(e)}}{=} m_s^{-\frac{2}{\alpha}} \frac{1}{\Gamma\left(\frac{2}{\alpha}\right)} \frac{\Gamma\left(m_s+\frac{2}{\alpha}\right)}{\frac{2}{\alpha}\Gamma\left(m_s\right)}  \nonumber \\ 
  & \stackrel{\mathrm{(f)}}{=} m_s^{-\frac{2}{\alpha}} \frac{\Gamma\left(m_s + \frac{2}{\alpha}\right)}{\Gamma\left(1 + \frac{2}{\alpha}\right) \Gamma\left(m_s\right)},
\end{IEEEeqnarray}
  where (a) follows from the introduction of new variable $n \triangleq l-j$, $\textstyle \binom{l}{j} = \binom{l}{l-j}$, and the change of the order of summations,
    (b) follows from $\textstyle \sum_{l=n}^k \binom{l}{n} = \binom{k+1}{n+1}$,
    (c) follows from $\textstyle \binom{k+1}{n+1} = \binom{k+1}{k-n}$ and the introduction of new variables $t \triangleq k+1$ and $s \triangleq k-n$,
    (d) follows from the formula of $\textstyle \sum_{s=0}^t (-1)^s \binom{t}{s} \left((t-s)x+y\right)_{(k)} = 0$ for any complex numbers $x$ and $y$ when $t>k$,
    (e) follows from the formula of $\textstyle \sum_{k=0}^{m-1} \frac{\Gamma\left(k-\beta\right)}{\Gamma\left(k+1\right)} = -\frac{\Gamma\left(m-\beta\right)}{\beta\Gamma\left(m\right)}$ for any real number $\beta$ \cite{Garrappa07_BinCoeff},
    and (f) follows from $\textstyle \Gamma(1+z)=z\Gamma(z)$.

  Furthermore, if $m_s=m_i=m$, from the definition of $\Delta_s(m_s,\alpha)$ and $\Delta_i(m_i,\alpha)$,
\begin{IEEEeqnarray}{ll}\label{eq:DeltaSIm}
  \frac{\Delta_s(m,\alpha)}{\Delta_i(m,\alpha)} = \frac{1}{\Gamma\left(1 - \frac{2}{\alpha}\right)\Gamma\left(1 + \frac{2}{\alpha}\right)} = \frac{\sin(2\pi/\alpha)}{2\pi/\alpha}.
\end{IEEEeqnarray}
\hfill\IEEEQED

\section{Proof of Proposition~\ref{prop:Fxi}} \label{app:proof:prop:Fxi}

\revbegin
  The first order derivative of $f_{\xi}(x)$ with respect to $x$ is given by
\begin{IEEEeqnarray}{l}\label{eq:FxiDerivative}
  \frac{d f_{\xi}(x)}{dx} = \frac{2}{\alpha} x^{-\frac{2}{\alpha}-1}\left(1+\frac{x}{\delta}\right)^{-1} u_\xi(x).
\end{IEEEeqnarray}
%
  Because $x>0$, $\textstyle \frac{2}{\alpha} x^{-\frac{2}{\alpha}-1}\left(1+\frac{x}{\delta}\right)^{-1}$ in \eqref{eq:FxiDerivative} is always positive, thus the sign of $\textstyle \frac{d f_{\xi}(x)}{dx}$ is only determined by that of $u_\xi(x)$.
  Consider the derivative of $u_\xi(x)$ given by
\begin{IEEEeqnarray}{l}\label{eq:UxiDerivative}
  \frac{d u_{\xi}(x)}{dx} = \frac{1}{\delta} \left( \left(\frac{\alpha}{2}-1\right) - \log\left(1+\frac{x}{\delta}\right) \right).
\end{IEEEeqnarray}
  From the assumption of $\alpha>2$ in Section~\ref{ssec:Model}, if $\textstyle 0 < x < \delta \left( \exp\left(\frac{\alpha}{2}-1\right) - 1 \right)$, $\textstyle \frac{d u_{\xi}(x)}{dx}>0$, i.e., $u_{\xi}(x)$ is increasing.
  Accordingly, as long as $\textstyle \lim_{x\rightarrow 0^+} u_{\xi}(x)>0$, $u_{\xi}(x)>0$ for $\textstyle 0 < x < \delta \left( \exp\left(\frac{\alpha}{2}-1\right) - 1 \right)$.
  Note that $\textstyle \lim_{x\rightarrow 0^+} \frac{\frac{\alpha}{2}\frac{x}{\delta}}{(1+\frac{x}{\delta})\log(1+\frac{x}{\delta})} = \lim_{x\rightarrow 0^+} \frac{\frac{\alpha}{2}}{1+\log(1+\frac{x}{\delta})} = \frac{\alpha}{2} > 1$ by the L'H\^{o}pital's rule.
  Thus, $\textstyle \lim_{x\rightarrow 0^+} u_{\xi}(x)>0$.
  By contrast, for $\textstyle x \geq \delta \left( \exp\left(\frac{\alpha}{2}-1\right) - 1 \right)$, $\textstyle \frac{d u_{\xi}(x)}{dx} \leq 0$, i.e., $u_{\xi}(x)$ is decreasing and eventually becomes negative as $x$ increases.
\revend
  Therefore, on the interval of $x>0$, $u_\xi(x)$ crosses zero once from positive to negative.
  That is, equation $u_\xi(x)=0$ has the unique solution on the interval of $x>0$, which is equal to $x^*$.
  From these results, $f_{\xi}(x)$ is increasing for $\textstyle 0 < x < x^*$ while decreasing for $\textstyle x \geq x^*$.
  Thus, $f_{\xi}(x)$ is maximized at $x=x^*$.

\revbegin
  The increase in $x^*$ with $\alpha$ follows from the monotonic increase of $\textstyle v_{\xi}(x)\triangleq  \left(1+\frac{\delta}{x}\right)\log(1+\frac{x}{\delta})$, because $x^*$ is equal to the solution of $\textstyle v_{\xi}(x) = \frac{\alpha}{2}$.
  This can be shown as follows:
  The derivative of $v_\xi(x)$ is given by $\textstyle \frac{d v_{\xi}(x)}{dx} = \frac{1}{x^2}\tilde{v}_{\xi}(x)$,
  where $\textstyle \tilde{v}_{\xi}(x) \triangleq \left(x-\delta\log(1+\frac{x}{\delta})\right)$.
  Note that $\textstyle \lim_{x\rightarrow 0^+} \tilde{v}_{\xi}(x) = 0$ and $\tilde{v}_{\xi}(x)$ is increasing with $x>0$ because $\textstyle \frac{d \tilde{v}_{\xi}(x)}{dx} = 1-(1+\frac{x}{\delta})^{-1} > 0$ when $x>0$.
\revend
  Therefore, $\textstyle \frac{d v_{\xi}(x)}{dx} > 0$ for $x>0$, and the solution of $\textstyle v_{\xi}(x) = \frac{\alpha}{2}$, i.e., $x^*$, increases with $\alpha$.
\hfill\IEEEQED

\revend

%
%

%
%
%


\bibliographystyle{IEEEtran}
\bibliographystyle{IEEEbib}
\bibliography{./PeerDiscovery_arXiv_IEEEabrv,./PeerDiscovery_arXiv_Ref}

\end{document}